\documentclass{birkjour}
\RequirePackage[all]{xy}

\usepackage[noadjust]{cite}

\usepackage{amsmath,amssymb,amsfonts}
\usepackage{graphicx}
\usepackage{xcolor}
\newtheorem{thm}{Theorem}

\newtheorem{lem}[thm]{Lemma}

\theoremstyle{definition}

\theoremstyle{remark}

\newcommand{\T}{{\footnotesize \text{T}}}

\begin{document}
%
%
%
%

\title{Conceptual Problems in Quantum Squeezing}
\author[B. Mielnik]{Bogdan Mielnik$^\dagger$}
\address{Department of Physics \\ CINVESTAV-IPN \\  A.P. 14-740 \\ Mexico City 07000, Mexico.}
\author[J. Fuentes]{Jes\'us Fuentes}
\address{LCSB\\University of Luxembourg\\Belvaux L-4367, Luxembourg}
\email{jesus.fuentes@uni.lu}
\date{\today}
\dedicatory{This manuscript was submitted subsequent to the unfortunate passing of Prof. Mielnik}

\begin{abstract}

In studies of quantum squeezing, the emphasis is typically placed more on specific squeezed states and their evolution rather than on the dynamical operations that could simultaneously squeeze a broader range of quantum states, regardless of their initial configuration. We explore new developments in this area, facilitated by gently acting external fields which might induce squeezing of the canonical observables $q$ and $p$ in charged particles. The extensive research in this field has yielded many valuable insights, raising the question of whether there is still room for significant new contributions to our understanding. Nonetheless, we present some exactly solvable instances of this problem, observed in symmetric evolution intervals. These intervals allow for the explicit determination of the temporal dependence of external fields necessary to generate the required evolution operators. Our findings are linked with a straightforward application of Toeplitz matrices, offering a more accessible description of the problem compared to the frequently employed Ermakov-Milne invariants.

\end{abstract}
\label{page:firstblob}
\maketitle

\section{Introduction}
\label{sec:intro}

In traditional control problems in quantum theory, two principal frameworks are often juxtaposed: Schr\"odinger’s picture, focusing on state evolution, and Heisenberg’s, emphasising the evolution of observables. Dirac robustly advocated for the merits of Heisenberg’s representation in a notable polemic article \cite{dirac65}. Regardless of one's stance on his critique, certain practical aspects of his arguments find relevance in contemporary quantum control challenges. This is particularly true for linear transformations of canonical variables $q$ and $p$ including squeezing, as described by Yuen \cite{yuen76}. Yuen's application demonstrates the potential for precise and rapid measurement of free particle position using twisted canonical observables \cite{yuen83}. His techniques in squeezing presaged the advent of quantum tomography \cite{mancini96,manciniar,asorey11}, which represents quantum states through tomographic images derived from the Radon transforms of canonical variables \cite{radon17,dunajski}. This field is currently experiencing a surge of innovative ideas and contributions \cite{Seth}. Some of the consequences of this standpoint for basic quantum mechanical ideas, though not immediate, are not negligible.

Quantum mechanics formalism posits that every self-adjoint operator in Hilbert space corresponds to a measuring apparatus, capable of an instantaneous assessment, collapsing the wave packet into one of its eigenstates. However, practical laboratory experience suggests a more challenging scenario where measurement is preceded by an auxiliary evolution, and the fundamental indeterministic decision occurs at the final moment.

Our focus is particularly on auxiliary operations such as amplification or squeezing, exploring the prospects of applying Thorne et al.'s concept of demolition-free measurements \cite{nondemolishing78,nondemolishing80}. It is important to note, however, that not all available techniques inspire unequivocal confidence. This includes issues related to decoherence \cite{bohm66}, instability \cite{haroche98}, and delayed choice \cite{wheeler84}. Recently, even entangled states and their radiation effects have come under scrutiny \cite{Lloyd12,ma12,dolev,price}.

For these reasons, we focus exclusively on soft control techniques that avoid abrupt changes and minimise radiative pollution. Although our problem is elementary, it is not entirely trivial. In existing theory, it is not always assured that every well-defined unitary operator can be dynamically realised, or at least approximated, by realistic motion generators. Hence, we propose a simple class of time-dependent Hamiltonians to investigate whether they can effectively generate true squeezing effects in massive particles. We regard variable external fields as the sole plausible source of this phenomenon. Consequently, we disregard formal results derived for time-dependent masses and material constants, among others. Nevertheless, the array of impressive existing results cannot be overlooked \cite{dodonov02}. We also consciously steer clear of the potential complexities of quantum field theory \cite{wheeler84,Lloyd12}. Thus, we aim to confine our approach to a purely quantum mechanical framework, focusing on slow (adiabatic) processes that minimise radiative pollution. For analytical convenience, our calculations were performed using dimensionless variables.

\section{Elementary Evolution Matrices}
\label{sec:duality}

Let us consider a pair of conjugate variables $q$ and $p$, representing the dimensionless canonical position and momentum, respectively. In this dimensionless framework, we define $\tau$ as a dimensionless time parameter, adopting units where the mass $m = 1$ and $\hbar = 1$, thus ensuring $[q, p] = i$. Consider an elastic force represented by a nonsingular, bounded $\beta(\tau)$. The evolution equations generated by the Hamiltonian $H(\tau)=p^2/2 + \beta(\tau)q^2/2$ imply exactly the same linear equations for either classical or quantum canonical variables. These are: $dq/d\tau= p(\tau)$ and $dp/d\tau =-\beta(\tau)q(\tau)$. Over any time interval $[\tau_0, \tau]$, this leads to an identical transformation of the classical or quantum canonical pair. Such transformations are expressed by the same family of $2\times2$ symplectic evolution matrices $u(\tau,\tau_0)$:
\begin{equation}
  \label{eq:2}
  \begin{pmatrix} 
    q(\tau) \\ p(\tau)
  \end{pmatrix} = u(\tau,\tau_0) 
  \begin{pmatrix} q(\tau_0)\\p(\tau_0)\end{pmatrix}; \quad u(\tau_0,\tau_0) = 1,
\end{equation}
determined  by the matrix equations
\begin{equation}
  \label{eq:3}
  \frac{\text{d}}{\text{d}\tau} u(\tau,\tau_0) = \Lambda(\tau)u(\tau,\tau_0); \quad \Lambda(\tau) = \begin{pmatrix}0&&1\\-\beta(\tau)&&0\end{pmatrix}.
\end{equation}
The reciprocity between the classical and quantum pictures does not end up here. It turns out that, in the absence of spin, each unitary evolution operator $U(\tau,\tau_0)$ in $L^2(\mathbb{R})$ generated by the Hamiltonian $H(\tau)$ is determined, up to a phase factor, by the canonical transformation that it induces. This is the consequence of the following simple lemma \cite{reed75,mielnik77,mielnik11,mielnik13}: 

\begin{lem} The family of the unitary operators  $U(\tau,\tau_0)$
describing the evolution generated by the quadratic Hamiltonians $H(\tau)$,  is determined, up to a constant phase factor, by the corresponding matrices  $u(\tau,\tau_0)$. Consequently, it is also determined by the corresponding classical trajectories.
\end{lem}

\begin{proof} 
It is sufficient to observe that if two unitary operators $U_1$ and $U_2$ induce the same transformation of the canonical variables, namely $U_1^{\dagger}qU_1 = U_2^{\dagger}qU_2$ and $U_1^{\dagger}pU_1 = U_2^{\dagger}pU_2$, then $U_1U_2^{\dagger}$ commutes with both $q$ and $p$. Consequently, it also commutes with any function of $q$ and $p$. Given that in $L^2(\mathbb{R})$ the functions of $q$ and $p$ generate an irreducible algebra, $U_1U_2^\dagger$ must be a $c$-number. Being unitary, it can only be a phase factor, $U_1U_2^\dagger = e^{i\varphi} \Rightarrow U_1 = e^{i\varphi}U_2$, where $\varphi \in \mathbb{R}$.
\end{proof}

Any two unitary operators that differ only by a $c$-number phase, even though they act differently on the state vectors, induce the same transformation of quantum states. Therefore, we shall consider them equivalent, $U_1\equiv U_2$. It immediately follows that the trajectories of the classical motion problem with quadratic $H(\tau)$ fully determine the evolution of pure or mixed quantum states $\rho = \rho^\dagger \geq 0$, $\operatorname{Tr}\rho = 1$, and, modulo equivalence, the entire unitary history, which we denote for simplicity as $U(\tau) = U(\tau,\tau_0)$. Our description complements the trends of phase geometry \cite{berry84,anandan90,fernandez94,brody03,wolf04,carlini06,berry09,muga11,wolf12}; while it does not describe geometric phases, it determines alternative aspects of quantum states like the motion of centres, packet shapes, and all nuances of the statistical interpretation.

\section{Classification of the Motion of Massive Particles}
\label{sec:optics}

In quantum optics of coherent photon states, the notion of parametric amplification, as discussed by Mollow and Glauber \cite{mollow67}, plays an important role. Yet, in the description of massive particles the Heisenberg evolution of the canonical observables---the trajectory picture---receives less attention, even though it allows for the extension of optical concepts \cite{wolf04,wolf12}. This is particularly interesting for charged particles in ion traps driven by time-dependent fields, whether or not they coincide with Paul's formula \cite{paul90}. The most interesting here is the case of quite arbitrary periodic potentials.

For a given periodic field $\beta$, where $\beta(\T + \tau) = \beta(\tau)$, the most important matrices \eqref{eq:2} are $u(\T +\tau_0, \tau_0)$ describing the repeated evolution events. Since they are symplectic, their algebraic structure is uniquely defined by a single number $\Gamma=\operatorname{Tr}u(\T +\tau_0, \tau_0) $. (We are not referring to  Ermakov-Milne invariants \cite{ermakov08,milne30,mayo02}.) Though the matrices $u(\T +\tau_0, \tau_0) $ depend on $\tau_0$, $\Gamma$ does not, permitting to classify the evolution processes generated by $\beta$ in any periodicity interval. The distinction between the three types of behaviour is quite elementary:

\begin{enumerate}
\item[I]  If $\vert \Gamma\vert < 2$, regardless of the details, the repeated $\beta$-periods produce an evolution matrix with a pair of eigenvalues $e^{i\sigma}$ and  $e^{-i\sigma}$ ($\sigma \in \mathbb{R}, 0<\sigma<\frac{\pi}{2}$) generating a stable (oscillating) evolution process. It allows the construction of the global creation and annihilation operators $a^+, a^-$ defined by the row eigenvectors of  $u(\T +\tau_0, \tau_0)$, thus characterising the evolution in the whole periodicity interval \cite{mielnik11,mielnik10}.

\item[II] If $\vert\Gamma\vert = 2$, the process generated by $\beta$ falls on the stability threshold with eigenvalues $\pm 1$ permitting to approximate a family of interesting dynamical operations, see also \cite{mielnik13}.

\item[III] If $\vert\Gamma\vert > 2$ then each single-periodic evolution matrix now has a pair of real, non-vanishing eigenvalues, $\lambda^+ = 1/\lambda^-$ with $\lambda^+ = e^\sigma$ and  $\lambda^-= e^{-\sigma}$ producing the squeezing of the corresponding pair of canonical observables $a^\pm$  defined again by the eigenvectors of  $u(\T +\tau_0, \tau_0)$, that is, $a^+$ expands at the cost of contracting $a^-$ or vice-versa.
\end{enumerate}

The above global data seem more relevant than the description in terms of the instantaneous creation and annihilation operators which do not make obvious the  stability thresholds. Specifically, for Paul potentials described by $\beta(\tau ) = \beta_0 + 2\beta_1 \cos \tau$, the map of the squeezing boundary is determined by the Strutt diagram  \cite{bender78}, traditionally limited to describe the ion trapping (in stability areas). Out of them are precisely the squeezing effects in III. To illustrate all this, integrating \eqref{eq:3} for the specific case of Paul's potential is insightful for $(\beta_0,\beta_1)$ out of the stability domain.

\section{Mathieu Squeezing}
\label{sec:mathieu}

\begin{figure}
\centering
\includegraphics[width=8cm]{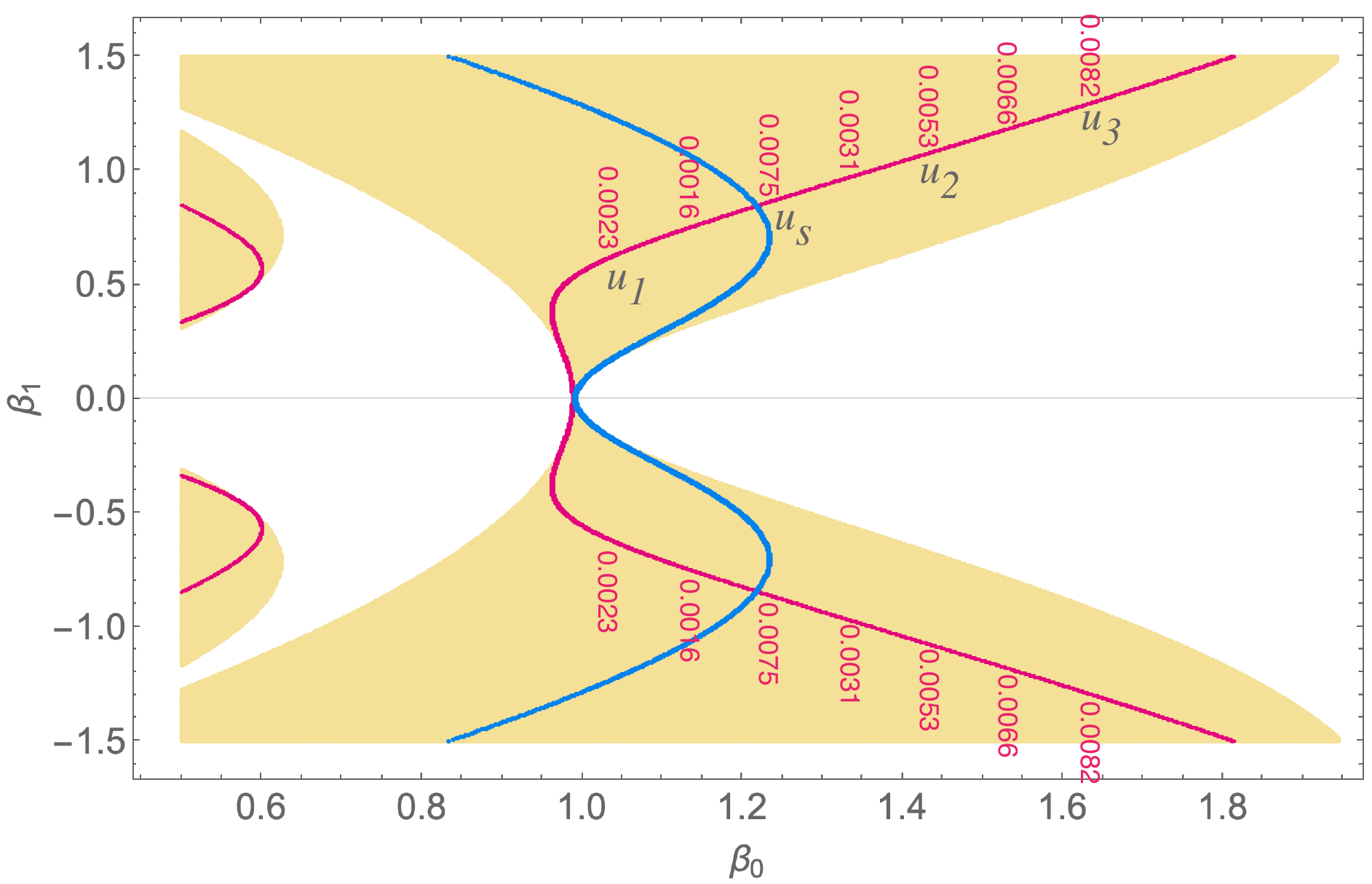}
\caption{In a Paul trap, squeezing is represented by the red curve, which plots the evolution matrices \eqref{eq:2} where $u_{12} \to 0$, within the squeezing region (coloured areas) on the Strutt map. Similarly, the blue curves denote the pairs $(\beta_0, \beta_1)$ for which $u_{21} \to 0$. Points located in the negative portions of the squeezing trajectories signify the occurrence of inverted squeezing effects.}
\label{fig:1}
\end{figure}

As established in previous studies \cite{wolf04, wolf12}, squeezing is unattainable when $\beta(\tau)$ exhibits symmetry within the operational interval. To investigate this further, we opted for numerical integration of Eq. \eqref{eq:3}. For this purpose, we employed a Paul potential, defined as $\beta = \beta_0 + 2\beta_1 \cos\tau$, over the interval $\left[\pi/2, 5\pi/2\right]$. We varied the parameters $\beta_0$ and $\beta_1$ in the squeezing region of the Strutt diagram, as reported in previous works \cite{bender78, mielnik10}.

Following this, we conducted a scanning procedure across regions I-III, with the aim of identifying those evolution matrices $u$ that produce quantum squeezing effects. Our findings, illustrated in Figure \ref{fig:1}, generalise the numerical data presented in \cite{mielnik10}. In the figure, the clear zones correspond to the stability domain (class I), while the coloured zones, shaded in yellow, correspond to the domain of unstable motion (class III), where squeezing effects occur. The boundaries between these clear and coloured areas furnish the class of motion II. The continuous red curves represent the combinations of $\beta_0$ and $\beta_1$ for which the matrix element $u_{12} = 0$ in the evolution matrix $u = u\left(5\pi/2, \pi/2\right)$. Conversely, the curves in blue outline the parameter sets $(\beta_0, \beta_1)$ such that $u_{21} = 0$. The intersection of both curves, denoted by $u_s$, yields the pair $(\beta_0,\beta_1)$ that grants the genuine position and momentum squeezing transformations, $q \rightarrow \lambda q$ and $p \rightarrow p/\lambda $, in that order.

The matrix elements $u_{12}$, represented by the values along the red line in Figure \ref{fig:1}, indicate subtle variations in the effects observed when $q$ is squeezed (or amplified). These variations specifically affect canonical variables, exemplified by $a^- = u_{21}q + 1/\lambda p$. To illustrate this, we selected four distinct evolution matrices. These matrices illustrate varying cases of squeezing, achieved through four different pairs of $(\beta_0,\beta_1)$, as follows:
\begin{equation}
\label{eq:9}
\begin{aligned}
u_1 &= \begin{pmatrix} 0.362 & \sim 0 \\ -1.114 & 2.751 \end{pmatrix}, &
u_2 &= \begin{pmatrix}0.175 & \sim 0\\ 3.501 & 5.798 \end{pmatrix},\\[4pt]
u_3 &= \begin{pmatrix}0.216 & \sim 0\\5.444 & 4.833\end{pmatrix}, & 
u_s&=\begin{pmatrix}0.227 & \sim 0\\ \sim 0 & 4.394\end{pmatrix}.
\end{aligned}
\end{equation}
Correspondingly, the specific matrices $u_1, u_2, u_3$, and $u_s$ were derived using the following sets of amplitude pairs:
\begin{equation*}
(\beta_0,\beta_1) = {(\pi/3,\pi/5), \, (\pi/2, 4\pi/13), \, (9\pi/16,5\pi/11), \, (4\pi/13,11\pi/41)}.
\end{equation*}

While the matrix $u_s$ at the intersection of both curves in the upper (positive) part of the diagram represents the coordinate squeezing $q \rightarrow \lambda q$, \   $p\rightarrow (1/\lambda) p$ with  $\lambda \approx 0.227$, the corresponding intersection on the lower (negative) part represents an inverse operation with  $\lambda^{-1} \approx 4.394$; thus, amplifying $q$ and squeezing $p$.  Henceforth, if the corresponding pulses were to be successively applied to two pairs of electrodes in a cylindric Paul trap, then the particle state would suffer a sequential expansion of its $x$ variable with the simultaneous squeezing of $y$, followed by the inverse amplification of $y$ and squeezing of $x$. It remains an open question whether new squeezing techniques could emerge from the generalisation of the operational techniques of high frequency pulses described in \cite{Itin}.

The reader might find the detailed focus on these minor aspects somewhat tedious. The boring problem of physical units brings, however, additional information. For the dimensionless time $\tau = \omega t$, consider the parameter $\T$. In the expression $\beta(\tau) = e\Phi(t)\T^2/r_0^2 m$, $\T$ represents the period of the oscillating Paul voltage on the trap wall, defined as $\Phi(t) = \Phi_0 + \Phi_1 \cos \omega t$. This leads to $\beta(\tau) = \beta_0 + 2\beta_1 \cos \tau$. Consequently, the same dimensionless matrix $u_s$ in \eqref{eq:9} is generated in $\left[\pi/2,5\pi/2\right]$ by physical parameters such that:
\begin{equation}
\label{eq:10}
\frac{e\Phi_0}{\omega^2 r_0^2m}=\beta_0, \quad \frac{e\Phi_1}{\omega^2 r_0^2m}=2\beta_1.
\end{equation}

For particles of a fixed mass and charge, the varying quantities are the potentials $\Phi$ and the physical time $\T=2\pi/\omega$ of the operations corresponding to the dimensionless interval $\left[\pi/2,5\pi/2\right]$. Hence, for any fixed $r_0$, the smaller $\omega$ is, the lower voltages $\Phi_0$ and $\Phi_1$ required to assure the same result, provided too weak fields do not permit the particle to escape or to collide with the trap surfaces. Consider a proton ($m = m_p \simeq \text{1.67} \times 10^{-24}$g) in an unusually large ion trap with $r_0 = 10$cm. If subjected to a moderately oscillating Paul field corresponding to a 3km long radio wave, one would find $\omega^2 r_0^2 m_p \simeq 10^{-12} \text{g}\frac{\text{cm}^2}{s^2} = \text{1.67} \times 10^{-12}\text{g}\frac{\text{cm}^2}{\text{s}^2} \simeq \text{1.04233}$eV, this leads to estimated voltages of: 
$\Phi_0 \simeq \text{1.0423}\beta_0 \text{V} \simeq \text{1.268}$V and 
$\Phi_1 \simeq \text{2.0846}\beta_1 \text{V} \simeq \text{1.759}$V. 
In an even larger trap, with $r_0 = 100$cm or, alternatively, for $r_0 =10$cm but the frequency 10 times higher, the voltages needed on the walls should be already 100 times higher!

While the analytic expressions in \cite{frenkel01} could yield more exact results, our computational experiments indeed indicate that the phenomena of  $qp$-squeezing can happen in Paul traps. However, these concern only the extremely clean oscillations, without any laser cooling, which are crucial for experimental trapping techniques \cite{paul90}, nor any dissipative perturbations. Moreover, the squeezing effects described by matrices \eqref{eq:9} are transient, manifesting only at precisely defined moments which makes difficult the observation of the phenomenon in the oscillating trap fields.

\section{The Option of Squeezed Fourier}
\label{sec:fourier}

Mathematically, one of the simplest ways to construct quantum operations is to apply sequences of external $\delta$-pulses that interrupt a continuous evolution process---such as the free evolution, the harmonic oscillation, or further interactions \cite{ammann97,fernandez92}. However, this method is limited by the practical challenge of applying $\delta$-pulses to external fields. In the case of squeezing, a more systematic method could involve composing evolution incidents which belong to the equilibrium zone I but their products do not. One possibility is to utilise segments of time-independent oscillator fields with elastic forces $\beta = \kappa^2 = \text{const.}$, generating symplectic rotations: 
\begin{equation}
\label{sprot}
u=\begin{pmatrix}\cos\kappa\tau & \frac{\sin\kappa\tau}{\kappa} \\ -\kappa \sin\kappa\tau & \cos\kappa\tau\end{pmatrix}.
\end{equation}
The simplest instances, obtained when $\cos\kappa\tau=0$, correspond to the squeezed Fourier transformations
\begin{equation} 
\label{sqF}
u=\begin{pmatrix}0&\pm\frac{1}{\kappa}\\ \mp \kappa&0\end{pmatrix}.
\end{equation}
Following the proposals in \cite{fan88} and \cite{grubl89} applying two such steps with different $\kappa$-values generate the evolution matrix:
\begin{equation}
\label{squeeze}
u_\lambda = \begin{pmatrix}0&\pm\frac{1}{\kappa_1}\\ \mp\kappa_1&0\end{pmatrix}\begin{pmatrix}0&\pm\frac{1}{\kappa_2}\\ \mp\kappa_2&0\end{pmatrix}=\begin{pmatrix}\lambda&0\\ 0&\frac{1}{\lambda}\end{pmatrix}; \quad \lambda = -\frac{\kappa_2}{\kappa_1}
\end{equation}
resulting in the squeezing of the canonical pair: $q \rightarrow \lambda q$, $p \rightarrow p/\lambda$, with the effective evolution operator: $U_\lambda = \exp[-i\sigma(pq+qp)/2]; \quad \sigma = \ln\lambda$. However, this requires two different values of $\kappa_1 \neq \kappa_2$ within two distinct time intervals separated by an abrupt potential jump. Here, the times $\tau_1$ and $\tau_2$ can fulfil $\kappa_1\tau_1 = \kappa_2\tau_2 = \pi/2$ to assure that both $\kappa_1$ and $\kappa_2$ grant two distinct squeezed Fourier operations in their time intervals. If one wants to apply two potential steps on the null background, it means at least three jumps ($0 \rightarrow \kappa_1 \rightarrow \kappa_2 \rightarrow 0$). The precise method of approximating a jump in the elastic potential remains a question. Moreover, each $\kappa$-jump indicates an energy transfer to the micro-particle \cite{grubl89}. Hence, could the pair of generalised Fourier operations in \eqref{squeeze} be superposed in a soft way with an identical final result? In fact, as we will demonstrate in the following section, recent advances in the inverse evolution problem reveal the existence of such effects.

\section{Toeplitz Matrices and Exact Operations}
\label{sec:soft}

While the exact expressions in \eqref{squeeze} were already known, it was previously overlooked that they can be produced using the simple anti-commuting algebra of $2\times2$ equidiagonal, symplectic matrices $u$, where $u_{11} = u_{22} = \frac{1}{2} \operatorname{Tr}u$. Remarkably, for any two such matrices $u$ and $v$, their anti-commutator $uv+vu$, as well as the symmetric products $uvu$ and $vuv$, are part of the same family.

The Toeplitz matrices, which have spurred considerable research \cite{bottcher00,trefethen05,Deift}, have not been widely recognised for their fundamental quantum control applications. In our context, \eqref{squeeze} does not require the elimination of jumps. This allows for greater versatility in constructing squeezed Fourier operations. These operations are achieved through symmetric products of multiple small symplectic contributions \eqref{sprot}, each operating over different time intervals with varying $\beta$ fields. Employing segments of symplectic rotations $v_k$, induced by Hamiltonians $H(\tau)$ with specific $\beta = \beta_k$ during time intervals $\Delta\tau_k ~ (k=0,1,2,\ldots)$, allows us to define the symmetric product:
\begin{equation}
\label{eq:12}
u=v_n\ldots v_1v_0v_1\ldots v_n
\end{equation}
which remains symplectic and equidiagonal---akin to the simplest Toeplitz class---with matrix elements $u_{11} = u_{22} = \frac{1}{2}\operatorname{Tr}u$. When \eqref{eq:12} achieves $\operatorname{Tr}u = 0$, the matrix $u$ transforms into a squeezed Fourier matrix. Their continuous counterparts are similarly obtainable.

To achieve this, one must assume that the amplitude $\beta(\tau)$ is symmetric around a certain point $\tau = 0$, thereby ensuring $\beta(\tau) = \beta(-\tau)$. Consequently, by considering the limits of small jumps $du$ caused by applying contributions $dv = \Lambda(\tau)d\tau$ from both sides, one derives the differential equation for $u = u(\tau,-\tau)$ within the expanding interval $[-\tau, \tau]$:
\begin{equation}
\label{eq:13}
\frac{du}{d\tau} = \Lambda(\tau) u+u\Lambda(\tau).
\end{equation}
This anti-commuting form leads easily to an exact solution. Since $\Lambda(\tau)$ is given by \eqref{eq:3}, equation \eqref{eq:13} becomes \cite{mielnik13,mielnik10}:
\begin{equation}
\label{eq:14}
\frac{du}{d\tau} = \begin{pmatrix}u_{21}-\beta u_{12} & \operatorname{Tr}u \\
-\beta\operatorname{Tr}u&u_{21}-\beta u_{12}
\end{pmatrix}
 =(u_{21}-\beta u_{12})\mathbb{I}+\operatorname{Tr}u\begin{pmatrix}0&1\\-\beta & 0\end{pmatrix}.
\end{equation}

For a given symmetric $\beta(\tau)$, the former equation explicitly determines the matrices $u = u(\tau,-\tau)$ for the expanded time interval $[-\tau,\tau]$ in terms of just one function $\theta(\tau) = u_{12}(\tau,-\tau)$. Indeed, as \eqref{eq:14} suggests the same differential equation for both $u_{11}$ and $u_{22}$, namely $du_{11}/d\tau=du_{22}/d\tau=u_{21}-\beta u_{12}$, and given that $u_{11} = u_{22} = 1$ at $\tau = 0$, it follows that $u_{11} = u_{22} = \frac{1}{2}\operatorname{Tr}u = \frac{1}{2}\theta'(\tau)$ for $u$ in any time interval $[-\tau,\tau]$. Furthermore, as $u = u(\tau,-\tau)$ matrices are symplectic, which means $\text{Det~} u = \left[\theta'(\tau)/2\right]^2-\theta u_{21} = 1$, one obtains 
\begin{equation}
\label{u21}
u_{21}=\frac{\left[\frac{1}{2}\theta'(\tau)\right]^2-1}{\theta}.
\end{equation}
Consequently, \eqref{eq:14} specifies the amplitude $\beta(\tau)$ required to generate the matrices $u=u(\tau,-\tau)$. Besides: 
\begin{equation}
\label{eq:15}
\beta u_{12} = u_{21} -\frac{du_{11}}{d\tau}.
\end{equation}
Given that $u_{12}=\theta$, and with $du_{11}/d\tau=\theta''/2$ along with $u_{21}$ as specified in \eqref{u21}, we can deduce that:
\begin{equation}
\label{beta} 
\beta=-\frac{\theta''}{2\theta}+\frac{\left[\frac{1}{2}\theta'(\tau)\right]^2-1}{\theta^2}.
\end{equation}

This solves the symmetric evolution problem for $u$ and $\beta$ in any interval $[-\tau,\tau]$ in terms of one, almost arbitrary function $\theta(\tau)$, restricted by non-trivial conditions at single points only. Hence, \eqref{beta} is indeed an exact solution of the inverse evolution problem, offering $\beta(\tau)$ in terms of the function $\theta(\tau) = u_{1,2}(\tau,-\tau)$ representing the evolution matrices for the expanding (or shrinking) evolution intervals $[\tau,-\tau]$. Note though that the dependence of $u(\tau, \tau_0)$ on $\beta(\tau)$ given by \eqref{beta} in any non-symmetric interval $[\tau, \tau_0]$ requires still an additional integration of \eqref{eq:3} between $\tau_0$ and $\tau$. Some simple algebraic relations of $\beta$ and $\theta$ are worth attention.

\begin{lem} Let $\beta(\tau)$ be defined as in \eqref{beta}, over an interval $[-T, T]$, where $T \in \mathbb{R}$. Further assume that $\theta(\tau)$ is a continuous function and, at least, three times differentiable over the same interval. Then, the following conditions ensure the continuity and differentiability of $\beta$, as well as the dynamical relations between $\theta$ and $\beta$:
  \begin{enumerate}
  \item[\textnormal{i.}] At any point $\tau$ where $\theta(\tau) = 0$,
    it must be the case that $\theta'(\tau) = \pm 2$.
  \item[\textnormal{ii.}] If $\theta'''(\tau) = 0$, then
    $\beta'(\tau) = 0$.
  \item[\textnormal{iii.}] At any point $\tau$ where
    $\theta(\tau) \neq 0$ but $\theta'(\tau) = 0$, Eq. \eqref{eq:14}
    for the interval $[-\tau, \tau]$ represents the squeezed Fourier
    transformation with $\beta(\tau)$ at the endpoints satisfying
    $\beta(\tau)\theta^2 = -\frac{1}{2} \theta''\theta - 1$.
  \end{enumerate}
\end{lem}

\begin{proof}
The proof follows directly from \eqref{eq:15}. Specifically, \eqref{eq:14} and the initial condition ensure that $u_{11} = u_{22} = \frac{1}{2}\theta'(\tau)$. Therefore, when $\theta' = 0$, both $u_{11} = u_{22} = 0$, leading to $u_{12} = b \neq 0$ and $u_{21} = -1/b$, which corresponds to the general form of the squeezed Fourier transformation outlined in Section \ref{sec:fourier}. Concurrently, \eqref{beta} simplifies to yield 
\[
\beta(\tau)\theta^2 = -\frac{1}{2}\theta''\theta - 1 \implies \beta(\tau)b + \frac{\theta''}{2} + \frac{1}{b} = 0.
\]
In particular, if $\theta''(\tau)b = -2$, then $\beta(\tau) = 0$. This completes the proof.
\end{proof}

A certain curious quid pro quo should be noted. Without entering into the phase problems \cite{chen10,muga11,guerrero11}, we focused on the simplest case of Toeplitz matrices \cite{bottcher00,trefethen05}. This approach, also employed in \cite{mielnik11,mielnik13,mielnik10}, resolves the inverse evolution problem for $\beta(\tau)=\kappa^2(\tau)$ in terms of $\theta(\tau)$, without relying on any auxiliary invariants. However, its purely comparative sense, should be stressed. For a fixed pair of canonical variables $q,p$ it does not give the causally progressing process of the evolution, but rather compares the evolution incidents in a family of expanding intervals $[-\tau, \tau]$. Should one wish to follow the causal development of the classical or quantum systems, the Ermakov-Milne equation \cite{ermakov08,milne30} might be useful. An interrelation between both methods waits still for an exact description. It is not excluded that anti-commutator algebras could also be helpful in some higher-dimensional canonical problems.

\section{Imperfections and Open Problems}
\label{sec:imperfections}

\paragraph{Troubles With Geometry} In numerous laboratories, techniques employed for maintaining and cooling ions are adequate for studying atomic structures, but they fall short in broader applications. The time-dependent oscillator potential is often created on a highly local scale---specifically, in the immediate vicinity of the central axis of a quadrupole trap, which may consist of merely four metal bars \cite{senko02}. The opportunity to extend the reach of oscillator fields for controlling unitary evolution arises either in conventional or cylindrical Paul traps with ideally hyperbolic surfaces, or within the interiors of cylindrical solenoids, provided that the operational area is sufficiently large. If this condition is met, the propagation of the controlling field within the trap introduces an additional challenge.

\paragraph{Control of Currents} While the latter addresses the spatial constraints, it leaves unanswered the question of inducing temporally variable yet spatially homogeneous surface currents independent of the axis coordinate $z$. In the static case, the magnetic field $\mathbf{B}$ within a solenoid is generated by a stationary current $\mathbf{J}$ circulating around the solenoid's surface. Integrating the magnetic field along a closed contour around the solenoid leads to the well-known expression $B = \frac{4\pi}{c} \frac{\Delta J}{\Delta z}$. However, generating time-dependent currents that are uniform across all surface sections, irrespective of $z$, remains a challenge. If a solenoid were constructed from a single spiral wire wrapped around its cylindrical surface and connected at both ends to a potential difference $\Phi(t)$, any minor alteration in $\Phi$ would propagate as a current pulse along the solenoid, thereby creating $z$-dependent fields instead of a quasi-static $B(t)$. A potential workaround might involve covering the cylindrical surface with shorter wires connected to a common variable voltage source, though this is not the only possible solution.

\paragraph{The Model of a Rotating Cylinder} An alternative approach is inspired by an example described by Griffiths \cite{griffiths99}. Consider a cylindrical surface made of a non-conducting material (such as glass) with a radius $R$, uniformly charged with a surface density $\sigma$. Each circular belt of 1cm height would thus carry a charge proportional to $R\sigma$. Such an experimental setup is feasible. If the cylinder has a radius $R = 20$cm and rotates at a frequency $\omega = 1\, \text{s}^{-1}$ around its axis $z$, with each 1cm horizontal belt carrying a charge of 1C, it would generate a homogeneous magnetic field $\mathbf{n} B$ with an intensity given by
\begin{equation}
\label{eq:28}
B = \frac{4\pi}{c}\omega R\sigma \approx 1.25\,\text{G},
\end{equation}
at least under the post-post-Newtonian approximation. By varying the angular velocity $\omega = \omega(t)$, one could create a practically homogeneous magnetic field $\mathbf{n}B(t)$ within a quasi-static environment. The efficacy of such techniques remains an open question.

\paragraph{Time Control} To induce a quantum state operation via variable fields, the micro-object must be exposed to these fields for a precisely defined interval, from the operation's commencement to its conclusion. This is especially critical for operations driven by time-dependent magnetic fields. The charged particle must enter the solenoid in a known initial state at the exact moment a Fourier squeezing operation begins and must be assessed at a specific instant after one or multiple field pattern applications. This synchronisation requirement is often overlooked in quantum control literature. One might consider a long, yet finite, solenoid. Inject a particle at an exact moment into the solenoid at a point $z = z_0$ with a velocity $v_z$. The particle's wave packet propagates and alters its form through the solenoid, eventually arriving at the other end during the time span of one or multiple squeezing operations. A measuring device, such as a photographic plate, records its position on a new orthogonal plane. Several errors are inevitable, arising from the granularity of the detection system and the Heisenberg uncertainty principle, which affects the time of flight. An intriguing feature is that if the squeezing operation amplifies coordinates while shrinking momenta, this is akin to a non-demolition measurement \cite{nondemolishing78, nondemolishing80}, albeit with increased precision in deducing the initial coordinates.

\paragraph{Neglected Perturbations} Our calculations focus solely on pure particle states evolving in slowly varying external fields, while neglecting potential interactions with any residual matter in ion traps or solenoids. We have also omitted the direct packet reflection or absorption by laboratory walls, thus leaving questions open regarding the necessary dimensions of the ion trap or solenoid, as well as the unresolved issue of a time operator, even in the context of flat surfaces. Other potential modifications by dissipative mechanisms such as those described by Lindblad, Gisin, and Percival have also been ignored \cite{lindblad76, gisin92, gisin99}. Several unresolved and perhaps contentious aspects in quantum mechanics remain. Despite these, one should not forget that our techniques are grounded in orthodox quantum mechanics, specifically the evolution matrices \eqref{eq:2}. These low-energy phenomena may be as fundamental to the field of quantum theory as their high-energy counterparts.

\subsection*{Acknowledgments}
We acknowledge the discussions with Dr. Jan Gutt from Physics Department of Polish Academy of Science. J.F. wishes to express profound gratitude to the late Prof. Mielnik, whose intellectual contributions during the period of 2015-2018 were instrumental in shaping the foundations of this work. It is with a sense of deep respect and acknowledgment that this manuscript is submitted after his decease.

\bibliographystyle{spmpsci}
\bibliography{manuscript} 

\end{document}